\documentclass[journal,draftclsnofoot,onecolumn,12pt,oneside]{IEEEtran}
\usepackage{amsmath,amsfonts}
\usepackage{algorithmic}
\usepackage{algorithm}
\usepackage{amssymb}
\usepackage{color}

\usepackage{float} 

\usepackage{array}
\usepackage[caption=false,font=normalsize,labelfont=sf,textfont=sf]{subfig}
\usepackage{textcomp}
\usepackage{stfloats}
\usepackage{url}
\usepackage{bm}
\usepackage{verbatim}
\usepackage{multirow} 
\usepackage{ntheorem}
\usepackage{graphicx}
\usepackage{cite}
\theorembodyfont{\upshape}
\theoremheaderfont{\it}
\qedsymbol{\square}
\theoremseparator{:}
\newtheorem{theorem}{\indent Theorem}
\newtheorem{lemma}{\indent Lemma}
\newtheorem*{proof}{\indent Proof}
\newtheorem{remark}{\indent Remark}
\newtheorem{corollary}{\indent Corollary}

\makeatletter

\newcommand{\Rmnum}[1]{\expandafter\@slowromancap\romannumeral #1@}
\makeatother
\begin{document}

\title{A Universal Framework of Superimposed RIS-Phase Modulation for MISO Communication}

\author{Jiacheng~Yao,~\IEEEmembership{Student Member,~IEEE}, Jindan~Xu,~\IEEEmembership{Member,~IEEE}, Wei~Xu,~\IEEEmembership{Senior Member,~IEEE}, Chau~Yuen,~\IEEEmembership{Fellow,~IEEE}, and Xiaohu~You,~\IEEEmembership{Fellow,~IEEE}
\thanks{Jiacheng Yao, Wei Xu, and Xiaohu You are with the National Mobile Communications Research Laboratory (NCRL), Southeast University, Nanjing 210096, China (\{jcyao, wxu, xhyu\}@seu.edu.cn).}
\thanks{Jindan Xu and Chau Yuen are with Engineering Product Development (EPD) Pillar, Singapore University of Technology and Design, Singapore 487372, Singapore (\{jindan\_xu, yuenchau\}@sutd.edu.sg).}}

%

\maketitle

\begin{abstract}
To fully exploit the additional dimension brought by reconfigurable intelligent surface (RIS), it is recently suggested by information theory that modulating information upon RIS phases is able to send extra information with increased communication rate. In this paper, we propose a novel superimposed RIS-phase modulation (SRPM) scheme to transfer extra messages by superimposing information-bearing phase offsets to conventionally optimized RIS phases. The proposed SRPM is interpreted as a universal framework for RIS phase modulation. Theoretical union bound of the average bit error rate (ABER) of the proposed SRPM is also derived with the maximum likelihood (ML) detection. The diversity order is characterized as $\frac{1}{2}$ for all parameter settings, which is useful for determining the optimal choice of the phase modulation parameters. Furthermore, we discover that doubling the number of either RIS reflecting elements or the transmit antennas is equivalent to a $3\,$dB increment in the transmit power for SRPM. Numerical results demonstrate the effectiveness of SRPM and reveal that it achieves reliable communication of more bits than existing schemes.
\end{abstract}


\section{Introduction}
\IEEEPARstart{T}{he} rapid growth of wireless data demands in recent years has put forward higher requirements for the sixth-generation (6G) mobile network. Reconfigurable intelligent surface (RIS), which has received worldwide attention, is regarded as an important enabling key technology for the 6G network \cite{6g}.  By integrating a large number of passive low-cost controllable reflecting elements, RIS can build smart radio environments, thereby greatly improving both spectral and energy efficiencies \cite{ref1,ref2,xujindan}.\par

In current mainstream of studies on RIS \cite{ref3,ref4,ref5}, passive beamforming has been intensively investigated to intelligently adjust the RIS phases aiming at signal-to-noise ratio (SNR) maximization. However, it is recently discovered that passive beamforming is suboptimal in utilizing the RIS and RIS has shown the potential for delivering additional information despite being a passive device. Index modulation (IM) is one of the promising candidate techniques, in which the additional information bits are modulated in the ON/OFF states of the transmission entities \cite{im1}. Recently, many studies have focused on the combination of IM and RIS to improve the spectral efficiency \cite{im1,ref18,im2,im3}. However, due to the limited number of transmission entities in practice, the performance gain brought by RIS-based IM schemes is limited, which do not fully unlock the potential of RIS.

On the other hand, the channel capacity analysis in \cite{ref6} for RIS-assisted networks suggests to bear additional information with the reflection phase control. Moreover, from the perspective of degree-of-freedom (DoF), a significant improvement with theoretical justification was observed  by modulating information in RIS phases\cite{ref7}. To realize RIS phase modulation, studies \cite{ref8,ref9,ref10,ref11} have proposed various schemes for extra information transfer through RIS. In \cite{ref8}, a passive beamforming and information transfer (PBIT) scheme was proposed, where the ON/OFF states of each reflecting element were used to deliver an additional binary message. To overcome high outage probability caused by the varying number of activated elements in PBIT, the authors in \cite{ref9} proposed an RIS-based reflection pattern modulation scheme, referred to as RIS-RPM, where the RIS elements are grouped into subsets and only the subset corresponding to the extra information to sent was activated at a time. 
In \cite{ref10}, an RIS-based quadrature reflection modulation (RIS-QRM) scheme was proposed with improved performance. Instead of switching ON/OFF the reflecting elements, the RIS-QRM scheme used I/Q phases while activating all the reflecting elements. However, the transmission rate of these extra bits is limited by the number of RIS subsets which is usually a small value. In \cite{ref11}, an uplink system with two users was considered, where the extra data of the second user was transmitted by rotating the RIS phases by a specific offset angle. Similarly in \cite{ref12}, a phase rotation scheme was proposed to transmit additional information by integrating the principles of spatial modulation (SM) and phase modulation.\par

However, most existing methods superimposed upto one extra bit per reflecting element via the RIS phase modulation. Meanwhile, extra information introduces additional randomness of RIS element inactivation and RIS phase changes, which inevitably leads to dynamic and severe SNR degradation. Moreover, to the best of our knowledge, most theoretical analysis, e.g., in terms of the average bit error rate (ABER), for RIS phase modulation schemes considered the single-input single-output (SISO) setting for sake the of tractability. The ABER analysis for multi-antenna systems is still missing. \par

To tackle the above issues, we proposed a novel phase modulation scheme by leveraging tunable phase offsets superimposed upon the conventionally optimized RIS phase for passive beamforming. This scheme is inspired by the fact that moderate phase noise, e.g., caused by discrete phase constraints, has marginal impact on the performance for large-size RIS communication systems\cite{ref13,ref14,ref15,discrete}. Therefore in the proposed scheme, we exploit offsets superimposed upon the conventional RIS phase shifts design for extra information transfer. We present a general framework for the RIS-based phase modulation that enables higher-order information modulation and achieves better performance. Furthermore, we present theoretical analysis of the ABER for the RIS-assisted multiple-input-single-output (MISO) system. Simulation results demonstrate the effectiveness of the proposed scheme and validate its superiority compared to existing schemes.\par
The rest of this paper is organized as follows. In Section~\Rmnum{2}, we formulate the system model and give the details of the proposed scheme. Section \Rmnum{3} derives the analytical performance of the proposed scheme. Simulation results and conclusion are given in Sections \Rmnum{4} and \Rmnum{5}, respectively.

\section{System Model}
We consider an RIS-aided downlink communication system, where a single-antenna user is served by an $N_t$-antenna base station (BS). Due to the existence of obstacles, no direct links between the BS and the user and an RIS with $N$ reflecting elements is deployed to help bridge the communication. Without loss of generality, let us group the reflecting elements uniformly into $L$ sub-surfaces to reduce the complexity of system design\cite{ref16}.  Each sub-surface contains $N/L$ elements, where $N/L$ is assumed as an integer for simplicity. Let $\bm{G}\in \mathbb{C}^{N\times N_t}$ and $\bm{h}_r \in \mathbb{C}^{N\times 1}$ denote the channel between BS and RIS and the channel between RIS and the user, respectively. The RIS is usually deployed where exists a strong line-of-sight (LoS) link to the BS. Accordingly, the channel between BS and RIS can be modelled as
\begin{equation}
\bm{G}= \beta \bm{{\rm{a}}}_N \left( \phi_r\right) \bm{{\rm{a}}}_{N_t}^H \left( \phi_t\right),
\end{equation}
where $\beta$ represents the complex gain of the path, and $\phi_r$ and $\phi_t$ denote the angle of arrival at RIS and the angle of departure at BS, respectively. The steering vectors $\bm{{\rm{a}}}_N \left( \phi_r\right)$ and $\bm{{\rm{a}}}_{N_t} \left( \phi_t\right)$ are respectively defined as
\begin{align}
\bm{{\rm{a}}}_N \left( \phi_r\right)&=\left [1,e^{j\frac{2\pi d}{\lambda}\sin \phi_r },\cdots,e^{j\frac{2\pi d}{\lambda}(N-1)\sin \phi_r }\right]^T,\nonumber \\
\bm{{\rm{a}}}_{N_t} \left( \phi_t\right)&=\left [1,e^{j\frac{2\pi d}{\lambda}\sin \phi_t },\cdots,e^{j\frac{2\pi d}{\lambda}(N_t-1)\sin \phi_t }\right]^T,
\end{align}
where $d$ is the antenna separation distance, $\lambda$ is the wavelength, and $\frac{d}{\lambda}$ is assumed as $\frac{1}{2}$. Moreover, for the channel between RIS and the user, due to rich scattering in the propagation environment, we model it as a Rayleigh channel. The $n$th element of $\bm{h}_r$, denoted by $h_{r,n}=\alpha_n e^{j\varphi_n}$, follows $\mathcal{CN}(0,1)$. The phase shift at the RIS is denoted by $\bm{\Theta}={\rm{diag}}\{e^{j\theta_1},\cdots, e^{j\theta_N}\}$, where $\theta_n$ is the phase of the $n$th reflecting element. Let $\bm{w}$ denote the normalized precoding vector at BS. The signal received at the user when RIS does not convey any information is formulated as
\begin{equation}
y=\sqrt{P} \bm{h}_r^H \bm{\Theta} \bm{G}\bm{w} s+z,
\end{equation}
where $P$ is the transmit power, $s$ is the symbol transmitted by the BS, and $z$ represents the additive Gaussian noise with zero mean and variance $\sigma^2$. The symbol $s$ is selected from $M$-ary phase shift keying/quadrature amplitude modulation (PSK/QAM) constellations, satisfying $\mathbb{E}\{\vert s\vert^2\}=1$. Note that RIS only reflects signals passively and it neither amplifies nor introduces noise \cite{pan}.

\par

To achieve higher communication rate, we propose the superimposed RIS-phase modulation, referred to as SRPM, where the RIS not only performs passive beamforming, but also implicitly conveys extra messages. Specifically, in the SRPM scheme, the phase configured at the RIS is made up of two parts, i.e., the base phases and the phase offsets. The base phase, which are used for enhancing the SNR at the receiver and avoiding high link outage probability, are selected by conventional optimization algorithms for the RIS passive beamforming \cite{ref3,ref4,ref5}. On the other hand, the dynamic phase offsets superimposed to the preciously optimized RIS phases are selected from a specific set, thereby enabling extra information transfer. Concretely, the phase on the $n$th reflecting element is configured as
\begin{equation}
\tilde{\theta}_n=\theta_n^* + k_n \Delta \theta,
\end{equation}
where $\theta_n^*$ is the optimized base phase, $k_n$ is selected from the set $\mathcal{K}=\{-K,-K+1,\cdots,0,1,\cdots,K\}$, and $\Delta \theta$ is a predetermined step unit of the phase offsets. Defining $K$ as the modulation order for conveying extra information at the RIS, the maximum phase offset is $K\Delta \theta$. Using the superimposed phase in (4), the received signal in (3) becomes
\begin{align}
y&=\sqrt{P} \bm{h}_r^H {\rm{diag}}\{e^{j\tilde{\theta}_1},\cdots, e^{j\tilde{\theta}_N}\} \bm{G}\bm{w}^* s+z\nonumber \\
&= \sqrt{P} \bm{h}_r^H \bm{\Xi}\bm{\Theta}^* \bm{G}\bm{w}^* s+z,
\end{align}
where $\bm{\Xi}\triangleq {\rm{diag}}\{e^{j k_1 \Delta \theta},\cdots, e^{j k_n \Delta \theta}\}$ conveys the additional information. By applying the optimized phase shifts derived in \cite{ref17}, we have
\begin{align}
\theta_n^* &=-\angle \left( [{\rm{diag}}(\bm{h}_r^H)\bm{{\rm{a}}}_N \left( \phi_r\right)]_n\right), \enspace n=1,2,\cdots,N,\nonumber \\
\bm{w}^*&=\frac{\left(\bm{h}_r^H \bm{\Theta}^* \bm{G}\right)^H}{\left \Vert \bm{h}_r^H\bm{\Theta}^* \bm{G} \right \Vert},
\end{align}
where $[\cdot]_n$ denotes the $n$th element of the vector, and $\bm{w}^*$ is determined by the maximum ratio transmission (MRT) principle\footnote{For general MISO systems, we can exploit the alternating optimization procedure to obtain the suboptimal phase shifts and precoding vector \cite{ref3}.}. Substituting $\bm{\Theta}^*$ and $\bm{w}^*$ into (5), we obtain 
\begin{align}
y&=\sqrt{P N_t}\beta \left( \sum_{n=1}^N x_n \alpha_n \right)+z,
\end{align}
where we define $x_n\triangleq e^{j k_n\Delta \theta} s$ containing all the information to transmit. At the receiver, we consider the maximum likelihood (ML) principle to detect $x_n$, which is expressed as
\begin{equation}
(\hat{s},\hat{\bm{v}})=\arg \min_{s,\bm{v}} \left \vert y-\sqrt{P N_t}\beta \left( \sum_{n=1}^N x_n \alpha_n \right)\right \vert^2,
\end{equation}
where $\bm{v}=\left [e^{j k_1 \Delta \theta},\cdots, e^{j k_n \Delta \theta} \right]^T$.\par
To reduce complexity, we consider that the reflecting elements in the same sub-surface are configured with the same phase offset. Then, the received signal in (7) is rewritten as
\begin{equation}
y=\sqrt{P N_t}\beta \left( \sum_{l=1}^L x_l\sum_{n\in\mathcal{A}_l} \alpha_n \right)+z=\sqrt{P N_t}\beta \bm{h}^T \bm{x}+z,
\end{equation}
where $\mathcal{A}_l$ represents the set of antenna elements in the $l$th sub-surface, $\bm{h}\triangleq \left[\sum_{n\in\mathcal{A}_1} \alpha_n,\cdots,\sum_{n\in\mathcal{A}_L} \alpha_n \right]^T$, and $\bm{x}\triangleq [x_1,\cdots,x_L]^T$. Now the data rate of the proposed SRPM scheme is calculated as $\log_2 M+L\lfloor\log_2(2K+1)\rfloor$ bits per channel use. Moreover, the ML detection algorithm requires $\mathcal{O}\left(MKL^2 \right)$ real multiplications, which shares the same complexity order-of-magnitude as the RIS-based IM techniques \cite{im3}. It is obvious that by reducing $L$, the complexity of ML detection is significantly reduced and the method is more computationally efficient than typical IM schemes.

\par

It is worth noting that the proposed SRPM scheme is a general framework for modulating extra information upon the phases of RIS and the previous works can be regarded as special cases of the SRPM. In particular, the proposed SRPM with $\mathcal{K}=\{0,1\}$ and $\Delta \theta=\pi/2$ reduces to the RIS-QRM in \cite{ref10}. For the case of discrete phase shifts with $b$ quantization bits, the proposed SRPM with $K=2^{b-1}$ and $K\Delta \theta=\pi$ is consistent with the method developed in \cite{ref6}. 

For the universal framework of SRPM, we note that a few special cases may not be applicable in practice. This is because some coincidental phase offsets can make the received symbol not uniquely decodable even for ideal channels without noise. For example, when $s=e^{j\frac{\pi}{4}}$ and $v_l=e^{-j\frac{\pi}{4}}$ for all $l$, we cannot distinguish it from another symbol $s=e^{-j\frac{\pi}{4}}$ and $v_l=e^{j\frac{\pi}{4}}$.

\section{Performance Analysis}
In this section, we analyze the ABER of the proposed SRPM scheme with the ML detection. The optimal selection of modulation parameters, $K$ and $\Delta\theta$, is presented with the diversity order of SRPM.\par
Firstly, we calculate the conditional pairwise error probability (CPEP) of detecting $(s,\bm{v})$ in (10) as $(\hat{s},\hat{\bm{v}})$. It follows
\begin{align}
{\rm{Pr}}(s,\bm{v} \to \hat{s},\hat{\bm{v}} \vert \bm {h})
&={\rm{Pr}}\left( \left\vert y-\sqrt{P N_t}\beta\bm{h}^T \bm{x}\right \vert^2 >\left\vert y-\sqrt{P N_t}\beta\bm{h}^T \hat{\bm{x}}\right \vert^2 \right)\nonumber \\
&={\rm{Pr}}\left( \left\vert z \right \vert^2 >\left\vert z+\sqrt{P N_t}\beta\bm{h}^T(\bm{x}- \hat{\bm{x}})\right \vert^2 \right)\nonumber  \\
&\overset{(a)}{=}\mathcal{Q} \left(\sqrt{\frac{P N_t\beta^2 \vert \delta\vert^2}{2\sigma^2}} \right)=\mathcal{Q} \left(\sqrt{\frac{P N_t \beta^2\lambda}{2\sigma^2}} \right),
\end{align}
where $\hat{\bm{x}}=\hat{s}\hat{\bm{v}}$, $\delta \triangleq \bm{h}^T(\bm{x}-\hat{\bm{x}})$, and $\lambda \triangleq \vert \delta \vert^2$. The equality in $(a)$ exploits the fact that CPEP is equivalent to ${\rm{Pr}}\left( -P N_t\beta^2 \vert \delta \vert^2 -2{\rm{Re}}\left\{\sqrt{P N_t} \beta \delta^* z \right\}>0\right)$ and $ -P N_t\beta^2 \vert \delta \vert^2 -2{\rm{Re}}\left\{\sqrt{P N_t} \beta \delta^* z \right\}$ is a Gaussian random variable with mean $ -P N_t\beta^2\vert \delta \vert^2$ and variance $2 P N_t \beta^2 \sigma^2 \vert \delta \vert^2$. Then, utilizing the definition of $\mathcal{Q}$ function in \cite[Eq. (2)]{ref19}, we arrive at this equality.

Then, the average pairwise error probability (APEP) is obtained as
\begin{align}
{\rm{Pr}}(s,\bm{v} \to \hat{s},\hat{\bm{v}})=\int_{0}^{+\infty} \mathcal{Q} \left(\sqrt{\frac{P N_t \beta^2 \lambda}{2\sigma^2}} \right) f_{\lambda} (\lambda) \mathrm{d}\lambda,
\end{align}
where $f_\lambda (\cdot)$ is the probability distribution function (PDF) of~$\lambda$. In order to get an insightful expression of the APEP in (11), we need to obtain the distribution of $\lambda$. To proceed, we resort to finding the moment-generating function (MGF) of $\lambda$ in \emph{Lemma~1}.

\begin{lemma}
The MGF of $\lambda$ can be expressed separately for the following two cases.
\begin{itemize}
\item $\det(\bm{C})\neq 0$: 
\begin{align}
&\mathcal{M}_\lambda(t)=\left(\det(\bm{I}-2t\bm{C})\right)^{-\frac{1}{2}} {\rm{exp}} \left(-\frac{1}{2}\bm{m}^T\left[\bm{I}-(\bm{I}-2t\bm{C})^{-1} \right] \bm{C}^{-1} \bm{m}\right),
\end{align}
\item $\det(\bm{C})= 0$:
\begin{align}
\mathcal{M}_\lambda(t)=(1-2c\sigma_x^2t)^{-\frac{1}{2}}\exp \left( \frac{c\mu_x^2 t }{1-2c\sigma_x^2t}\right),
\end{align}
\end{itemize}
where $\det(\cdot)$ is the determinant of the matrix, and $\bm{C}$, $\bm{m}$, $c$, $\mu_x$, and $\sigma_x$ are defined in Appendix A, and the MGFs are obtained for asymptotically large $N/L$.
\end{lemma}
\begin{proof}
Please refer to Appendix A. $\hfill\square$
\end{proof}

Applying the derived MGF of $\lambda$, we can calculate the APEP according to the Gil-Pelaez's inversion formula in \cite{ref18}. However, the complex calculations also require the aid of numerical calculation tools, which hardly gives useful insights. To circumvent this difficulty, we further exploit a useful approximation of the $\mathcal{Q}$ function by $\mathcal{Q}(x)\approx\frac{1}{12}e^{-\frac{x^2}{2}}+\frac{1}{4}e^{-\frac{2x^2}{3}}$ \cite[Eq. (14)]{ref19}. Then, the APEP in (11) becomes
\begin{align}
\Pr(s,\bm{v} \to \hat{s},\hat{\bm{v}})
&\approx \int_{0}^{+\infty} \left(\frac{1}{12}e^{-\frac{P N_t\beta^2\lambda}{4\sigma^2}}+\frac{1}{4}e^{-\frac{P N_t\beta^2\lambda}{3\sigma^2}} \right) f_{\lambda} (\lambda) \mathrm{d}\lambda \nonumber \\
&=\frac{1}{12}\mathcal{M}_\lambda \left(-\frac{P N_t\beta^2 }{4\sigma^2} \right)+\frac{1}{4}\mathcal{M}_\lambda \left(-\frac{P N_t\beta^2}{3\sigma^2} \right).
\end{align}

Now that, according to this derived APEP, we obtain a union bound of the ABER as follows
\begin{align}
P_b\leq& \frac{1}{M (2K+1)^L}\sum_{s} \sum_{\bm{v}} \sum_{\hat{s}} \sum_{\hat{\bm{v}}}{\rm{Pr}}(s,\bm{v} \to \hat{s},\hat{\bm{v}}) \frac{e(s,\bm{v} \to \hat{s},\hat{\bm{v}})}{\log_2(M)+L\lfloor\log_2(2K+1)\rfloor},
\end{align}
where $e(s,\bm{v} \to \hat{s},\hat{\bm{v}})$ represents the number of bits in error for the corresponding pairwise error event. For each misestimated $\hat{s}$, the number of error bits is upto $\log_2M$, and for each misestimated $\hat{v}_l$, the number of error bits is upto $\lfloor\log_2(2K+1)\rfloor$. Thus, $e(s,\bm{v} \to \hat{s},\hat{\bm{v}})$ is calculated by accumulating all these error bits. In addition, we analyze the diversity order of the proposed SRPM scheme in \emph{Theorem 1}.\par

\begin{theorem}
The diversity order of the proposed SRPM is equal to $\frac{1}{2}$ for all available parameter settings of $K$ and $\Delta \theta$.
\end{theorem}
\begin{proof}
Please refer to Appendix B. $\hfill\square$
\end{proof}

\begin{remark}
Compared with typical RIS reflection schemes in \cite{ref9,ref10}, the proposed SRPM achieves the same diversity order. Moreover, note that no matter how the modulation parameters $K$ and $\Delta\theta$ are selected, the diversity order remains unchanged. However, the choices of $K$ and $\Delta\theta$ do have an impact on the ABER performance. In other words, we can optimize the choices of $K$ and $\Delta\theta$ for different requirements. Especially for discrete phase shifts, an exhaustive search with low complexity is able to get the optimal choice. 
\end{remark}

\begin{corollary}
By direct inspection of (14), it is found that doubling transmit antennas, $N_t$, is essentially equivalent to doubling the transmit power from the perspective of improving the ABER performance, which results in a theoretical gain of 3$\enspace$dB. Furthermore, by substituting (13) into (14) and combining the definition of $\sigma_x^2$ in Appendix A, it reveals that doubling RIS reflecting elements is also equivalent to doubling the transmit power. Therefore in the RIS deployment, the required power consumption can be reduced by deploying more low-cost reflecting elements.
\end{corollary}

\begin{corollary}
When there is only one sub-surface at RIS, i.e., $L=1$, the MGF of $\lambda$ in (12) and (13) reduces to
\begin{align}
\mathcal{M}_{\lambda}(t)=\left(1-\frac{(4-\pi)\lambda Nt }{2}\right)^{-\frac{1}{2}}\mathrm{exp}\left(\frac{\lambda \pi N^2t}{4-(8-2\pi)\lambda Nt  } \right).
\end{align}
Then, at low SNR regime, the MGF in (16) is dominated by the exponential term, and thus the ABER declines exponentially with SNR. Moreover, a much lower ABER can be achieved by increasing $N$ due to the exponent $N^2$ in this exponential~form.

\end{corollary}

\section{Simulation Results}
In this section, simulation results are provided to verify the effectiveness of the SRPM. The received SNR is defined as the ratio of transmit power and noise power, i.e., $\frac{P}{\sigma^2}$.
Unless otherwise specified, the parameters are set as: the number of reflecting elements, $N=128$, the number of transmit antennas at the BS, $N_t=8$, the number of sub-surfaces, $L=2$, the modulation order $K=1$, and the step unit of phase offsets, $\Delta\theta=\frac{3\pi}{16}$. QPSK is used at the BS. For simplicity, the path loss is normalized to $\beta=1$.

\begin{figure}[!t]
	\centering
	\begin{minipage}[t]{0.33\linewidth}
		\centering
		\includegraphics[width=1\linewidth]{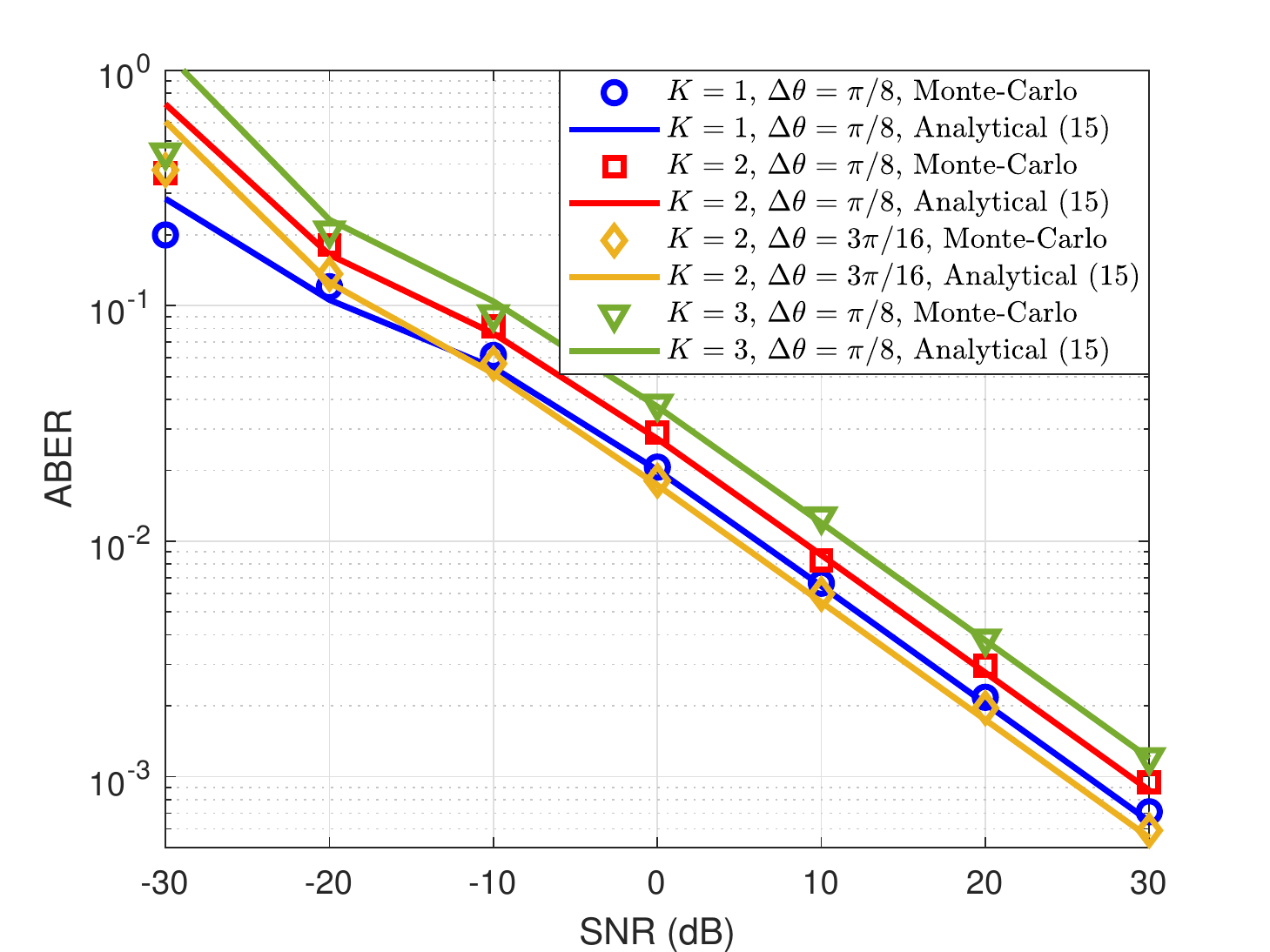}
	\end{minipage}
	\begin{minipage}[t]{0.33\linewidth}
		\centering
		\includegraphics[width=1\linewidth]{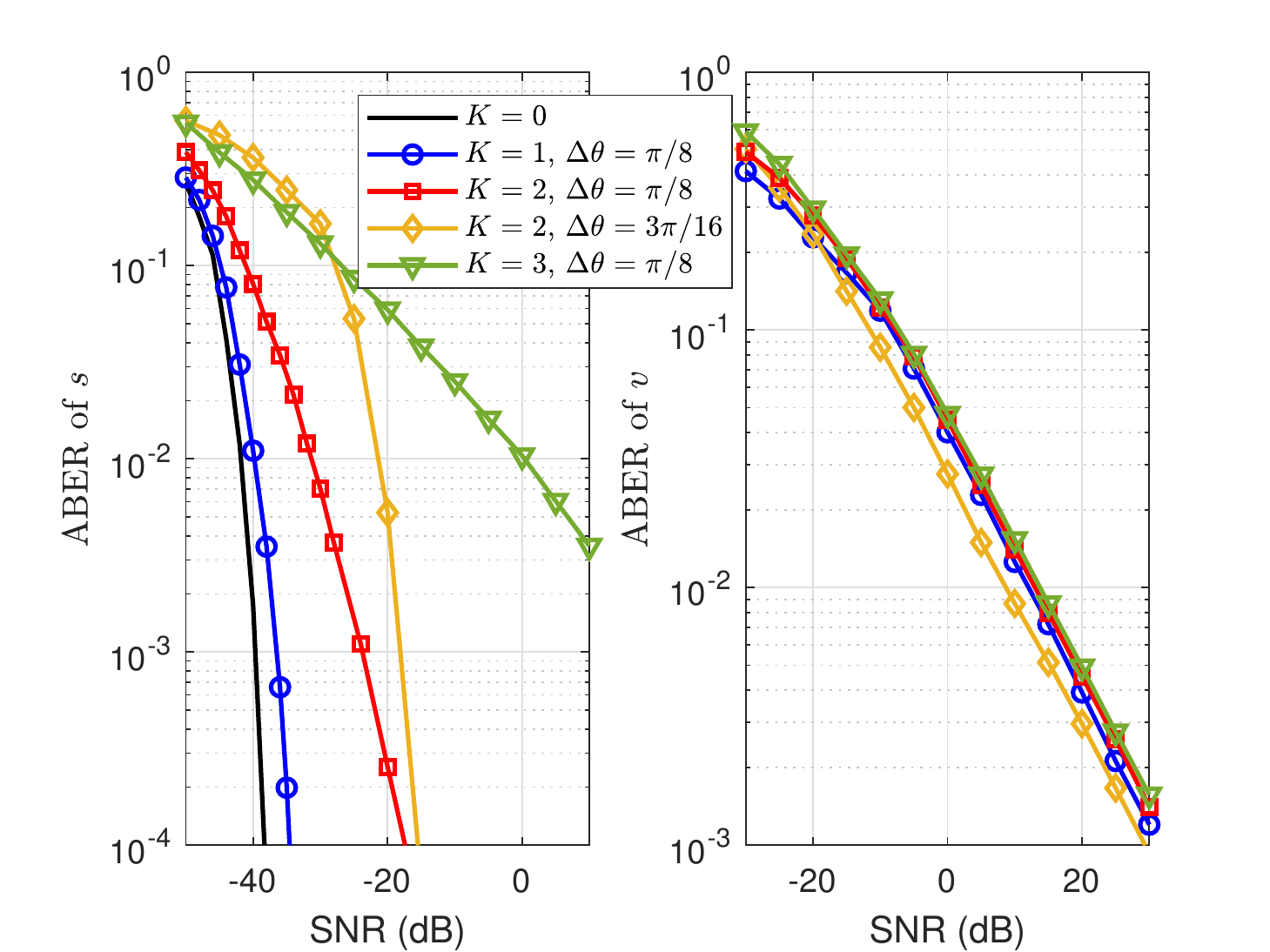}
	\end{minipage}
	
	\caption{ABER of the SRPM scheme under different $K$ and $\Delta\theta$. (a) ABER of $(s,\bm{v})$. (b) ABER of $s$ and ABER of $\bm{v}$.}
\end{figure}

Fig. 1(a)  depicts the numerical ABER obtained from Monte Carlo simulations, compared with theoretical approximation in (15) under various choices of $K$ and $\Delta\theta$. It is seen that under high SNR, the theoretical approximation in (15) is tight for all the tested setups. Moreover, given $\Delta\theta$, the performance of ABER gradually deteriorates as $K$ increases. This is because the larger $K$ results in more constellation symbols and smaller distance between different constellation symbols. In addition, for fixed $K$, we observe a performance gain when $\Delta\theta$ increases from $\frac{\pi}{8}$ to $\frac{3\pi}{16}$. This is because the distance between different constellation symbols grows with the phase offset. However, excessive phase offset causes a severe drop in the SNR at the receiver, which further affects the ABER performance. Hence, for given $K$, it is necessary to optimize $\Delta\theta$ in pursuit of better performance.\par

Fig. 1(b) depicts the ABER curves of messages $s$ and $\bm{v}$ separately to further show the ABER degradation for message $s$.  We find that different selections of $K$ and $\Delta \theta$ mainly affect the ABER of message $s$ and, as expected, it has little impact on  the ABER of message $\bm{v}$. The degradation of the ABER of message $s$ is mainly due to the maximum phase offset, i.e., $K\Delta\theta$, which affects the received power. In other words, the transmission of  extra bits is at the cost of the degradation of the ABER of message $s$.

Regarding the optimal choice of $\Delta \theta^*$, we first list the values in Table I. Note that for the case of discrete phase shifts with $b=4$ quantization bits, we obtain the optimal $\theta^*$ by first optimizing $\theta_n$ for the continuous phase shift case according to (6) and then choose the closest phase  in the discrete phase shift set as $\theta_n^*$. Next, we obtain the optimal $\Delta\theta^*$ with the lowest ABER via exhaustive search. It can be found that the optimal $\Delta\theta^*$ decreases as $K$ increases. That is, with lower modulation orders, larger phase offsets are preferred, and vice versa. For the modulation mode of the symbols transmitted by the BS, it has marginal impacts on the optimal $\Delta\theta^*$. \par

\begin{table*}[!t]
\renewcommand\arraystretch{1.05}
\caption{Optimal Choices of $(K,\Delta\theta)$\label{tab:table1}}
\centering
\begin{tabular}{|c|c|c|c|c|c|c|c|c|}
\hline
\multicolumn{2}{|c|}{Parameters} & $K=1$ &$K=2$ &$K=3$&$K=4$ &$K=5$&$K=6$ &$K=7$\\
\hline
\multirow{2}{*}{ optimal $\Delta\theta^*$} &QPSK &${7\pi}/{16}$ &${3\pi}/{8}$ &${3\pi}/{16}$ &${3\pi}/{16}$ &${3\pi}/{16}$ &${\pi}/{16}$ &${\pi}/{16}$\\
\cline{2-9}
&16QAM&${3\pi}/{8}$ &${3\pi}/{8}$ &${3\pi}/{16}$ &${3\pi}/{16}$ &${3\pi}/{16}$ &${\pi}/{16}$ &${\pi}/{16}$\\
\hline
\end{tabular}
\end{table*}

\begin{figure*}[!t]
	\centering
	\begin{minipage}[t]{0.32\linewidth}
		\centering
		\includegraphics[width=1\linewidth]{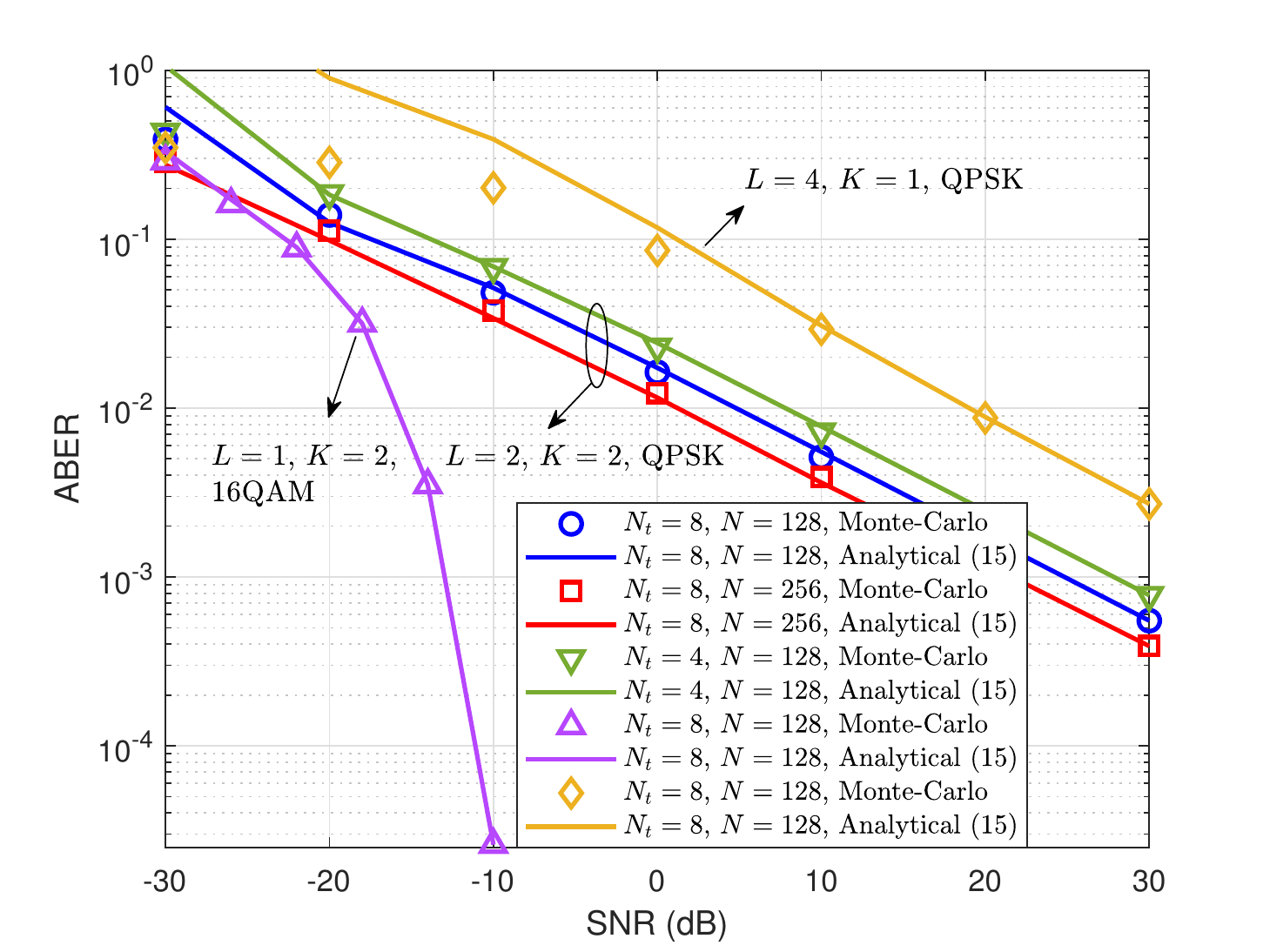}
		\caption{ABER of the SRPM scheme under different system parameters.}
	\end{minipage}
	\begin{minipage}[t]{0.32\linewidth}
		\centering
		\includegraphics[width=1\linewidth]{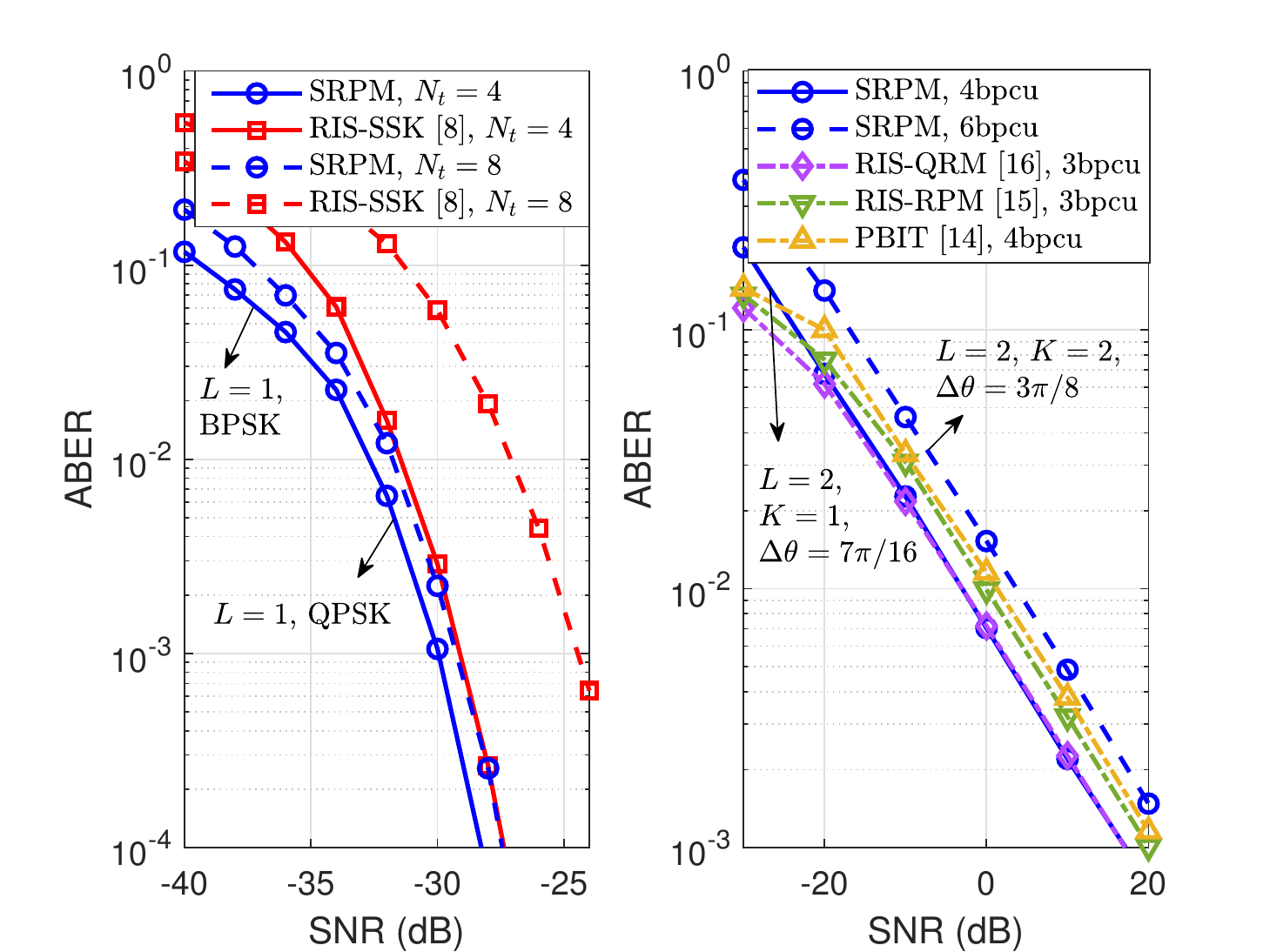}
		\caption{ABER of different schemes.}
	\end{minipage}
	\begin{minipage}[t]{0.32\linewidth}
		\centering
		\includegraphics[width=1\linewidth]{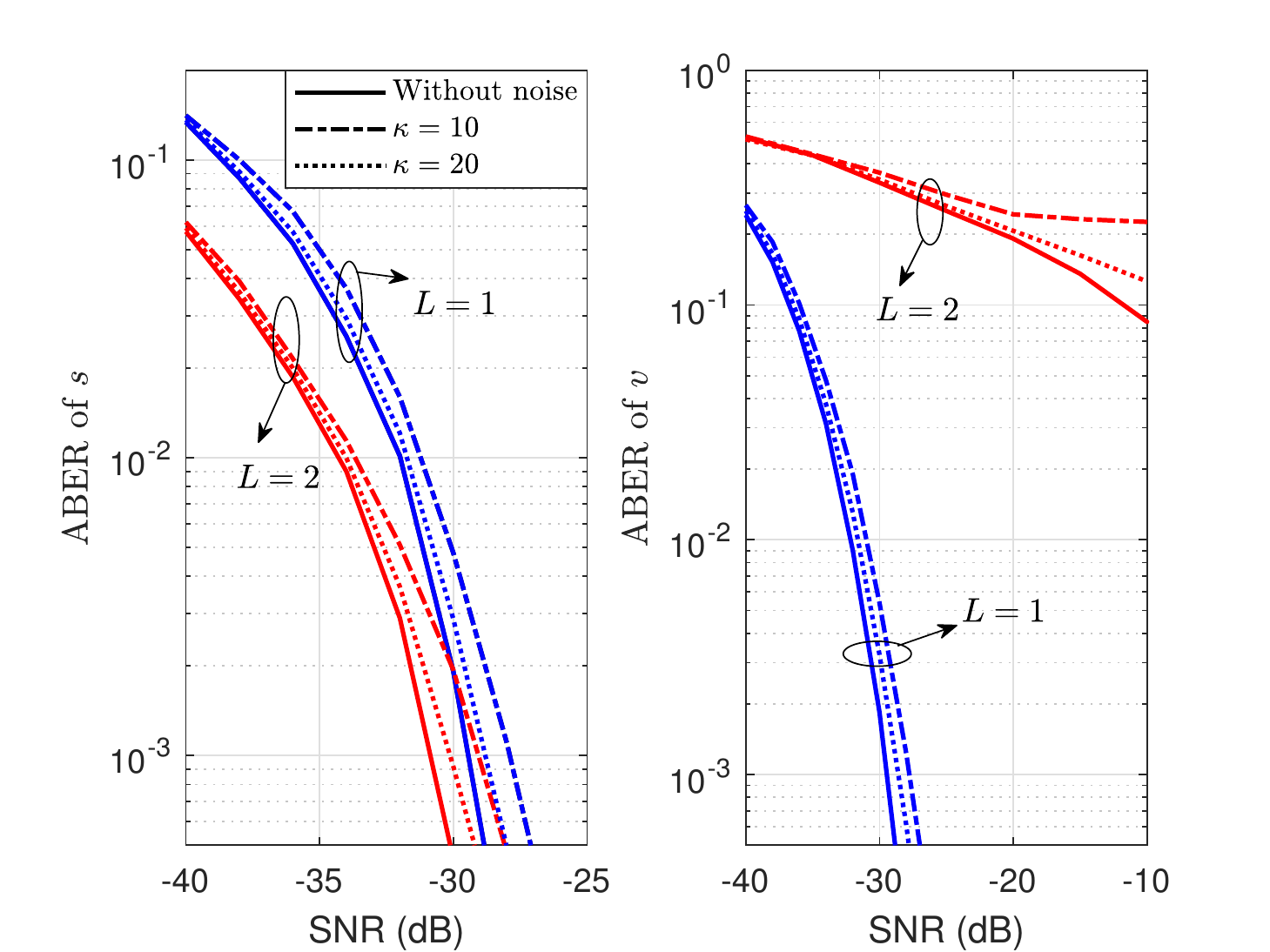}
		\caption{ABER of the SRPM scheme with phase noise.}
	\end{minipage}
\end{figure*}


In Fig. 2, we consider the impact of system parameters. Consistent with the theoretical analysis, doubling the reflecting elements or transmit antennas does bring a gain of 3$\enspace$dB. However, when the number of sub-surfaces, $L$, is increased, the ABER degrades significantly. Hence, to convey more information, it is preferred to increase $K$ than to increase $L$. Furthermore, note that the approximation is not tight enough with large $L$. This is because $N/L$ decreases and then makes then approximation in (18) inaccurate.\par

Then, we compare the ABER performance of the SRPM with that of state-of-the-art schemes in Fig.~3. The four benchmark schemes are described as: 1) The RIS-SSK scheme where only one transmit antenna is activated for extra information transfer \cite{im1}; 2) the PBIT scheme where sub-surfaces are randomly turned on or off with equal probability [14]; 3) the RIS-RPM scheme where one of the sub-surfaces is turned off [15]; 4) the RIS-QRM scheme where the phase offsets of $L=2$ sub-surfaces are set as $0$ and $\pi/2$, respectively [16]. It is observed from Fig. 3 that the proposed SRPM outperforms all benchmarks in terms of ABER. Moreover, the proposed SRPM has the advantage in realizing higher-rate transmissions.


Finally, we consider the inevitable phase noise in practice and  evaluate the sensitivity of the SRPM to the phase noise in Fig. 4.  We assume that the phase noise at the RIS follows the von Mises (circular normal) distribution with concentration parameter $\kappa$ \cite{ref18}. From Fig. 4, the performance loss in terms of SNR is about 1 dB and 2 dB for message $s$ when $\kappa=20$ and $10$, respectively. While for the phase-modulated message $\bm{v}$, the degradation in ABER is insignificant when there is only one sub-surface. However, for the SRPM with $L=2$, it is observed that the ABER deteriorates noticeably and there  exhibits error floor at high SNRs. This is because multiple sub-surfaces are not in phase and possible phase noise can impose a severe impact on the messages modulated on these RIS phases, especially for high SNR regimes.

\section{Conclusion}
In this letter, we proposed a novel SRPM scheme for conveying extra information, which is interpreted as a universal framework for modulating information in the phases of RIS. Analytical results show that SRPM can achieve a diversity order of $\frac{1}{2}$ for arbitrary parameters and doubling reflecting elements is equivalent to a $3\,$dB increment in the transmit power. Simulation results confirm that SRPM transmits more data and achieves lower ABER and further show that adding possible offsets is a more efficient way to transfer additional information than increasing the number of sub-surfaces. Moreover, the development of low-complexity detection algorithms for the SRPM remains an interesting future research direction.

\appendices
\section{Proof of Lemma 1}
To begin with, we express $\delta$ as
\begin{align}
\delta=\bm{h}^T(\bm{x}-\hat{\bm{x}})=\sum_{l=1}^L h_l (x_l-\hat{x}_l),
\end{align}
where $h_l=\sum_{n\in\mathcal{A}_l} \alpha_n$. Considering that there are sufficient reflecting elements in each sub-surface, it is safe to approximate $h_l$ for $l=1,\cdots,L$ by a Gaussian random variable via the Central Limit Theorem (CLT) for large $\vert \mathcal{A}_l \vert$. Accordingly, their mean and variance are
\begin{align}
\mathbb{E}\{h_l\}=\frac{\sqrt{\pi}N}{2L}\triangleq \mu_h, \enspace \mathbb{V}\{h_l\}=\frac{(4-\pi)N}{4L}\triangleq \sigma_h^2.
\end{align}
Let $d_{l,r}$ and $d_{l,i}$ denote the real and imaginary parts of $x_l-\hat{x}_l$, respectively. We can express the real and imaginary parts of $\delta$, denoted by $\delta_r$, and $\delta_i$, as $\sum_{l=1}^L d_{l,r}h_l$ and $\sum_{l=1}^L d_{l,i}h_l$, respectively. Given that $\{h_l\}_{l=1}^L$ are approximated as  i.i.d. Gaussian variables, $\delta_r$ and $\delta_i$ also follow the Gaussian distribution with
\begin{align}
\mathbb{E}\{ \delta_r \}&=\sum_{l=1}^L d_{l,r} \mu_h\triangleq \mu_r,\enspace \mathbb{V}\{\delta_r\}=\sum_{l=1}^L d_{l,r}^2 \sigma_h^2\triangleq \sigma_r^2,\nonumber\\
\mathbb{E}\{ \delta_i \}&=\sum_{l=1}^L d_{l,i} \mu_h\triangleq \mu_i,\enspace \mathbb{V}\{\delta_r\}=\sum_{l=1}^L d_{l,i}^2 \sigma_h^2\triangleq \sigma_i^2.
\end{align}
However, $\delta_r$ and $\delta_i$ may not be independent of each other. Then,let us define $\bm{c}\triangleq[\delta_r,\delta_i]^T$. We express its mean vector, $\bm{m}$, and covariance matrix, $\bm{C}$, as 
\begin{align}
\bm{m}=[\mu_r,\mu_i]^T,\enspace \bm{C}=\left[\begin{array}{cc}
\sigma_r^2&\sigma_{12}^2 \\ \sigma_{12}^2&\sigma_{i}^2
\end{array}\right],
\end{align}
where $\sigma_{12}^2$ is defined as
\begin{align}
&\sigma_{12}^2\triangleq \mathbb{E}\{\delta_r \delta_i\}- \mathbb{E}\{\delta_r\}\mathbb{E}\{\delta_i\}\nonumber \\
&=\sum_{l=1}^L d_{l,r}d_{l,i} (\mu_h^2+\sigma_h^2)+\sum_{l=1}^L \sum_{m\neq l}^L d_{l,r}d_{m,i} \mu_h^2 -\mu_r\mu_i.
\end{align}
Now, we get $\lambda=\bm{c}^T\bm{c}$. When $\det(\bm{C})= 0$, it is easily checked that at least one of $\delta_r$ and $\delta_i$ is zero or the two random variables are linearly dependent. Hence, $\lambda$ is expressed as $c\vert \delta_{x}\vert^2$, where $c$ is a constant, $x\in\{r,i\}$, and $\delta_x\neq 0$. Now, we find that $\lambda$ follows the non-central chi-square distribution and the MGF of $\lambda$ is
\begin{align}
\mathcal{M}_\lambda(t)=(1-2c\sigma_x^2t)^{-\frac{1}{2}}\exp \left( \frac{c\mu_x^2 t }{1-2c\sigma_x^2t}\right).
\end{align}
For the other case when $\det(\bm{C})\neq 0$, according to \cite[Eq.~(31)]{ref18}, we calculate the MGF of $\lambda$ as
\begin{align}
\mathcal{M}_\lambda(t)&=\left(\det(\bm{I}-2t\bm{C})\right)^{-\frac{1}{2}} {\rm{exp}} \left(-\frac{1}{2}\bm{m}^T\left[\bm{I}-(\bm{I}-2t\bm{C})^{-1} \right] \bm{C}^{-1} \bm{m}\right).
\end{align}
The proof completes.

\section{Proof of Theorem 1}
The diversity order is defined as the slop of the ABER at high SNR values, which is analytically evaluated as \cite{do}
\begin{align}
D=\lim\limits_{P\to \infty} -\frac{\mathrm{log}_2 P_b}{P}.
\end{align}
Plugging (15) into (24), we obtain
\begin{align}
D=\lim\limits_{P\to \infty} -\frac{\log_2 {\rm{Pr}}(s,\bm{v} \to \hat{s},\hat{\bm{v}})}{\log_2 P}.
\end{align}
Then, for the case that $\det(\bm{C})\neq0$, we use (12) to get
\begin{align}
\lim\limits_{P\to \infty} -\frac{\log_2 {\rm{Pr}}(s,\bm{v} \to \hat{s},\hat{\bm{v}})}{\log_2 P}=1.
\end{align}
On the other hand, for the case that $\det(\bm{C})=0$, we have
\begin{align}
\lim\limits_{P\to \infty} -\frac{\log_2 {\rm{Pr}}(s,\bm{v} \to \hat{s},\hat{\bm{v}})}{\log_2 P}=\frac{1}{2}.
\end{align}
For example, when $v_l=e^{j\Delta\theta}$ and  $\hat{v}_l=e^{-j\Delta\theta}$ for $l=1,2,\cdots,L$, we get $\det(\bm{C})=0$. Therefore, the second case in (27) always exists. Then we conclude that the diversity order of the proposed SRPM is equal to $\frac{1}{2}$ and complete the proof.

\end{document}